\pdfoutput=1
\documentclass[10pt,pra,aps,superscriptaddress,nofootinbib,twocolumn,longbibliography]{revtex4-1}
\usepackage{amsmath,bbm}
\usepackage{amsthm}
\usepackage{amsmath}
\usepackage{latexsym}
\usepackage{amssymb}
\usepackage{float}
\usepackage{quantikz}
\usepackage{tikz}
\usepackage{graphicx}           
\usepackage{color}
\usepackage{xcolor}
\usepackage{mathpazo}
\usepackage{comment}
\usepackage{enumitem}
\usepackage{multirow}
\usepackage{setspace}
\usepackage{subcaption}
\usepackage[colorlinks=true,linkcolor=blue,citecolor=magenta,urlcolor=blue]{hyperref}
\usepackage{svg}

\newcommand{\be}{\begin{equation}}
\newcommand{\ee}{\end{equation}}
\newcommand{\bea}{\begin{eqnarray}}
\newcommand{\eea}{\end{eqnarray}}
\def\squareforqed{\hbox{\rlap{$\sqcap$}$\sqcup$}}
\def\qed{\ifmmode\squareforqed\else{\unskip\nobreak\hfil
\penalty50\hskip1em\null\nobreak\hfil\squareforqed
\parfillskip=0pt\finalhyphendemerits=0\endgraf}\fi}
\def\endenv{\ifmmode\;\else{\unskip\nobreak\hfil
\penalty50\hskip1em\null\nobreak\hfil\;
\parfillskip=0pt\finalhyphendemerits=0\endgraf}\fi}

\newcommand{\tr}{\text{Tr}}
\newcommand{\I}{\mathbbm{1}}

\makeatletter
\newtheorem*{rep@theorem}{\rep@title}
\newcommand{\newreptheorem}[2]{%
\newenvironment{rep#1}[1]{%
 \def\rep@title{#2 \ref{##1}}%
 \begin{rep@theorem}}%
 {\end{rep@theorem}}}
\makeatother

\newtheorem{thm}{Theorem}
\newreptheorem{thm}{Theorem}
\newtheorem{lemma}{Lemma}
\newtheorem{definition}{Definition}

\newtheorem{obs}{Observation}

\begin{document}

\title{Nonlocality without entanglement for multipartite quantum measurements}


\author{Satyaki Manna}
\author{Debashis Saha}
\affiliation{Department of Physics, School of Basic Sciences, Indian Institute of Technology Bhubaneswar, Odisha 752051, India}


\begin{abstract}
We show that a pair of two-qubit product measurements can be perfectly distinguished when accessed jointly, yet cannot be perfectly distinguished when their local components are distributed between distant parties restricted to local quantum operations and classical communication (LOCC). This establishes a measurement analogue of nonlocality without entanglement that, unlike its state-based counterpart, emerges in a remarkably simple setting. We further uncover a hierarchy of LOCC discrimination strategies: adaptive protocols, in which the probe used at one site may depend on classical information obtained from another, can strictly outperform restricted non-adaptive protocols, as demonstrated by a set of six bipartite qutrit product measurements. We also show that LOCC distinguishability of bipartite measurements can be intrinsically directional, depending on which party initiates the protocol. More generally, we construct pairs of $n$-qubit product measurements that are globally perfectly distinguishable but cannot be perfectly distinguished across any bipartition, even when all parties on either side of the bipartition are allowed arbitrary joint quantum operations. Finally, we show how quantum data hiding based on quantum-process nonlocality without entanglement can be realized using elementary optical components.
\end{abstract}

\maketitle


\section{Introduction}
\emph{Nonlocality Without Entanglement} \cite{Benett} is one of the most astonishing and unique signatures of quantum theory. This phenomenon occurs when a set of orthogonal, multipartite product states cannot be perfectly distinguished using local operations and classical communication. It highlights a striking manifestation of quantum nonlocality, where the composite system exhibits properties that cannot be inferred from its individual subsystems, even though no entanglement is present. Owing to its fundamental significance and wide-ranging applications in quantum information theory, this phenomenon has been the subject of extensive investigation over the years \cite{benett2,Walgate,Walgate2,watrous05,Heng,ha2021quantum,He2024,Halder_prl,cohen2,nathanson2005distinguishing,childs2013,chitambar,supic,manna2026nonlocality,fan2,Shi2022}. 

States, transformations, and measurements constitute the basic
operational ingredients of quantum theory. It is therefore natural to
ask whether nonlocality without entanglement (NLWE) is a property peculiar to quantum states, or whether it can arise intrinsically at the level of quantum processes. While discrimination of quantum transformations, particularly unitary processes, has been studied under various spatial and temporal constraints \cite{Manna_pra_2026,entangledstatesuseful,bavaresco2022unitary,duan,acin,acin2,imam2026,manna2026,Manna_2026,lien2026,bavaresco,nidhin,Piani,Nuradha_2025,vishal,111,Mishra_2024}, here we address this question for quantum
measurements.

Remarkably, we show that a pair of two-qubit product measurements can
be perfectly distinguished when accessed jointly using an entangled
probe, yet cannot be perfectly distinguished when their local
components are distributed between distant parties restricted to arbitrary LOCC. This establishes a measurement analogue of NLWE. In contrast to the state-based phenomenon, which requires a collection of five qutrit-qutrit product states \cite{benett2,Walgate2}, the measurement analogue already emerges for a \emph{pair} of product measurements in the smallest bipartite Hilbert space. 

The active nature of measurements, however, introduces features with
no direct counterpart in state discrimination. In particular, local measurement discrimination permits adaptive probing: the probe chosen at one site may depend on classical information obtained from another. We show that this freedom generates a strict hierarchy of local
discrimination strategies. A set of six bipartite qutrit product
measurements can be perfectly distinguished by an adaptive protocol but not when such adaptive probing is forbidden. Moreover, local distinguishability can be intrinsically directional, that is, whether a set of
bipartite measurements is perfectly distinguishable may depend on which party initiates the protocol.

We finally ask whether measurement-based NLWE can be genuinely multipartite, in the sense that it persists across every bipartition even after all parties on each side are allowed to act jointly \cite{Halder_prl}. We answer this question affirmatively by constructing, for arbitrary $n$, \emph{pairs} of $n$-qubit product measurements that are perfectly distinguishable globally but cannot be perfectly distinguished across any bipartition. This provides a particularly strong separation between global and distributed access to otherwise elementary product measurements.

Beyond its foundational implications, this process-based perspective offers a simple route to quantum data hiding, where classical information can be hidden directly in the configuration of robust macroscopic devices, such as optical elements, rather than in fragile quantum states, providing an experimentally accessible application of NLWE.



We begin by revisiting the problem of distinguishing a set of quantum states and measurements, and then formalize the notion of global and LOCC protocols.

\section{Distinguishability of quantum states}

\begin{definition}[Perfect discrimination]
\label{def:dist}
Let $\{\rho_i\}_i$ be a set of quantum states and let
$\{q_i\}_i$ be a set of non-negative weights.
We denote by
$\mathcal{P}\big[\{\rho_i\}_i,\{q_i\}_i\big]
$ a protocol that perfectly distinguishes the states $\{\rho_i\}_i$
with weights $\{q_i\}_i$. That is, there exists a POVM $\{M_i\}_i$, with
$M_i\geq 0, \sum_i M_i=\mathbb{I},$
such that for every $i$ with $q_i>0$,
$\operatorname{Tr}(\rho_i M_i)=1.$
\end{definition}
Note that, the way this distinguishability is defined here, we do not impose $\sum_i q_i=1$, i.e., $\{q_i\}_i$ may not be a probability distribution, on purpose. In the course of the paper, we will see that, measurement distinguishability will be related to the problem of distinguishability of the states such that the associated coefficients $\{q_i\}_i$ may not form a probability distribution. Thus, not taking $\{q_i\}_i$ as a probability distribution will help us to simplify the expression of distinguishability of quantum measurements.

Similarly, we define the discrimination protocol for multipartite states when the parties are restricted to local operations and classical communication.
\begin{definition}[Perfect LOCC discrimination]
\label{def:locc-dist}
Let $
\rho_i=\bigotimes_{k=1}^{N}\rho_i^{(k)}$ be a set of multipartite product states, and let $\{q_i\}_i$ be a
set of non-negative weights. We denote by
$\mathcal{P}_{\mathrm{LOCC}}\left[\left\{
\bigotimes_{k=1}^{N}\rho_i^{(k)}\right\}_i,\{q_i\}_i \right]$
an LOCC protocol that perfectly distinguishes these states. That is, the protocol produces a classical guess $\hat{i}$ such that
$p(\hat{i}=i\,|\,\rho_i)=1$ for every $i$ with $q_i>0$.
\end{definition}

\section{Distinguishability of quantum measurements}
Quantum measurement having $m$ distinct outcome is defined by a set of operators,
$
M :=\{F_a\}_a = \{F_1,\cdots,F_m\} , 
$ 
such that $\sum_{a=1}^m F_a^\dagger F_a = \I$. When this measurement is performed on a quantum state $\rho$,
the probability of obtaining outcome $a$ is given by Born rule as follows,
$
p(a|\rho,F_a) = \tr(\rho F^\dagger_aF_a) ,
$
and post-measurement state is,
$
\rho_a=\frac{F_a\rho  F^\dagger_a}{\tr(\rho F^\dagger_aF_a)}.
$

Now, we define distinguishability of quantum measurements. We consider a priori known set of $n$ measurements acting on $d$-dimensional quantum states, each having $m$ outcomes, defined by $\{F_{a|x}\}_{a,x}$, which are sampled from a  equal probability distribution $\{p_x\}_x$, i.e., $p_x > 0,  p_x =1/n,\forall x$. Here $x\in\{1,\cdots,n\}$ denotes the measurements and $a\in\{1,\cdots,m\}$ denotes the outcomes. We define the positive operators, commonly known by POVM elements, as
$
M_{a|x} = F^\dagger_{a|x}F_{a|x} ,$
which is often convenient to use as we do not consider the post-measurement states.  To distinguish these $n$ measurements $\{F_{a|x}\}_{a,x}$, the measurement device is given a known entangled state, and the device carries out one of these $n$ measurements. After performing one of the measurements, the device gives a reduced state with classical outcome. To distinguish $d$-dimensional measurements, it is sufficient to consider qudit-qudit entangled state (Theorem $4$ of \cite{manna_25}).

Distinguishability of the set of quantum measurements with entangled systems finally comes out as the distinguishability of the reduced state of Bob $\rho^B_{a|x}$ with the associated classical probability $p(a|x,\rho) = \tr \left[(M_{a|x}\otimes \I)\rho_{AB} \right] $, where $\rho_{AB}$ is the distinguishing probe. The reduced state of Bob when outcome $a$ is obtained by Alice for measurement $x$, 
$
\rho^B_{a|x} = \tr_A\left[\frac{\left(F_{a|x}\otimes \I\right)\rho_{AB}\left(F^\dagger_{a|x}\otimes \I\right)}{\tr\left( \rho_{AB} \left( M_{a|x} \otimes \I \right) \right)}  \right].
$ 
Depending on Alice's outcome, Bob performs a measurement described by the set of POVM elements $\{N_{b|a}\}$ where $b\in\{1,\cdots,n\}$ is the outcome and $a\in\{1,\cdots,m\}$ is the choice of measurement settings, $\{N_{b|a}\}\geqslant 0, \sum_b N_{b|a}=\I$. For distinguishability, this protocol is successful when Bob's output $b$ will be the same as Alice's input $x$. So, the success probability of this task depends on the joint probability of both Alice's and Bob's measurements. 

From Ref. \cite{manna_25}, it can be shown that the distinguishability of quantum measurements can be expressed as distinguishability of the states $\{\rho_{a|x}^B\}_x$ with weights $\{p\left(a|x,\rho\right\}_x$. Therefore, the following definition arrives.

\begin{definition}[Perfect discrimination]\label{dis_meas}

Let $\{F_{a|x}\}_{a,x}$ be a set of quantum measurements sampled from equal probability distribution.
We denote by
$\mathcal{D}\big[\{F_{a|x}\}_{a,x},\{\frac1n\}_x\big]
$ a protocol that perfectly distinguishes the measurements $\{F_{a|x}\}_{a,x}$, which means there exists $\rho_{AB}$ such that $\mathcal{P}\big[\{\rho_{a|x}^B\}_x,\{p\left(a|x,\rho\right\}_x\big]$ exists for every $a$.
\end{definition}
 If one chooses a product probe, it can be trivially checked that the reduced state on Bob's side does not give any information. Only necessary information can be provided from the classical outcome of Alice. For more details on discrimination of measurements, one can check \cite{manna_25,njp2021,Puchala,eid2026}.

\emph{ Global distinguishability (GD)---}
In this section, we present general protocols for global distinguishability of $v$ number of $w$-partite measurements $\{F_{\vec{a}|i}\}_{i=1}^v =\{\otimes_{k=1}^w F_{a_k|i}^{(k)}\}_{i=1}^v$. In general, the dimension of the measurements can be different for different parties. In this work, we only deal with the measurements which are locally implementable, i.e., the measurements are product.

In the global strategy for distinguishing a set of measurements,
the user can sequentially choose different subsystems of the $w$-partite measurement to discriminate them. Suppose the user first selects an arbitrary subset of parties, $\mathbbm{S}_1 \subseteq [w]$, where $[w] \equiv \{1,\cdots,w\}$, consisting of $w_1$ subsystems with $w_1\leq w$. 
 Thus, without loss of generality, the user prepares an arbitrary composite quantum state $\sigma_{AB}$ on the composite Hilbert space $\mathcal{H}_A \otimes \mathcal{H}_B$, where $\mathcal{H}_A = \bigotimes_{k\in\mathbbm{S}_1} \mathcal{H}_k $ represents the $w_1$-partite subsystem on which the measurements act. The unknown measurement $\otimes_{k\in\mathbbm{S}_1}F^{(k)}_{a_k|i}$ is applied to $A$ subpart of $\sigma_{AB}$ resulting in an reduced state at $B$ side $\sigma_{\vec{a}_{\mathbbm{S}_1}|i}^B= \tr_A\left[\frac{\left(\otimes_{k\in\mathbbm{S}_1}F^{(k)}_{a_k|i}\otimes \I\right)\sigma_{AB}\left(\otimes_{k\in\mathbbm{S}_1}F^{(k)\dagger}_{a_k|i}\otimes \I\right)}{p(\vec{a}_{\mathbbm{S}_1}|i,\sigma)}  \right]$ with classical probability $p(\vec{a}_{\mathbbm{S}_1}|i,\sigma) = \tr \left[(\otimes_{k\in\mathbbm{S}_1} M^{(k)}_{a_k|i}\otimes \I)\sigma_{AB} \right] $. If there exists $\mathcal{P}\big[\{\sigma^B_{\vec{a}_{\mathbbm{S}_1}|i}\}_i,\{p\left(\vec{a}_{\mathbbm{S}_1}|i,\sigma\right\}_i\big]$ for every $\vec{a}_{\mathbbm{S}_1}$, protocol is terminated here. Otherwise, the user employs any arbitrary channel on those measurement to extract some information. Based on this partial information, the user selects another subset $\mathbbm{S}_2 \subseteq [w]\setminus \mathbbm{S}_1$ of size $w_2$, prepares an appropriate probing state, and attempts to distinguish the corresponding $w_2$-partite measurements. At this stage, the final distinguishing measurement can be performed on the joint state obtained from the reduced states associated with $\mathbbm{S}_2$ and the transformed states resulting from the general channel applied to the previous subset. If the measurements are successfully distinguished, the protocol terminates; otherwise, the elimination procedure is continued. This process is repeated until all subsystems on which the multipartite measurement acts have been exhausted. Since there are multiple possible choices of subsets at each stage, each sequence of chosen subsets induces a partition of $[w]$. Furthermore, for a fixed partition, the protocol depends on the order in which the subsets are selected. Therefore, the user considers all possible strategies generated by all permutations of subsets over all partitions and, at each stage, performs a joint measurement to perfectly distinguish the given measurements.


\emph{LOCC distinguishability with restricted strategy (LDR)--- }
In the restricted class of local strategies, each party independently prepares a state $\rho^{(k)}_{A_kB_k}$ on $\mathcal{H}_{A_k} \otimes \mathcal{H}_{B_k}$ and lets the subsystem $B_k$ reduce under the unknown measurement. After a measurement $F^{(k)}_{a_k|i}$ acts on the respective subsystems, the parties got evolved states $\left\{\bigotimes_{k=1}^{w} \rho^{B_k}_{a_k|i} =\bigotimes_{k=1}^{w}\tr_A\left[\frac{\left(F^{(k)}_{a_k|i}\otimes \I\right)\rho^{(k)}_{A_kB_k}\left(F^{(k)\dagger}_{a_k|i}\otimes \I\right)}{p(a_k|i,\rho)}  \right] \right\}_{i=1}^v$ for each $\vec{a}$, corresponding to the given set of measurements $\{F_{\vec{a}|i}\}$ with relevant classical outcomes $p(a_k|i,\rho^{(k)})=\tr \left[( M^{(k)}_{a_k|i}\otimes \I)\rho^{(k)}_{A_kB_k} \right]$. Their tasks, therefore, boil down to find $\mathcal{P}_{\mathrm{LOCC}}\left[\left\{
\bigotimes_{k=1}^{w} \rho^{B_k}_{a_k|i}\right\}_i,\{\prod_{k=1}^w p(a_k|i,\rho^{(k)})\}_i \right]$.

\emph{LOCC distinguishability with adaptive strategy (LDA)--- }
In general, local strategies can be more intricate than the restricted ones described above. In an adaptive LOCC protocol, $k$-th party prepares its probing state $\rho^{(k)}_{A_k B_k}$ based on classical information about the outcome received from the previous party who have already performed measurements on their probed systems. Such a protocol involves many possible sequences in which the parties participate.
Say, the protocol begins with an initial party, say $k$, who prepares a state $\rho^{(k)}_{A_kB_k}$ and applies $ F^{(k)}_{a_k|i}$ from the unknown measurement $ F_{\vec{a}|i}$. On the reduce state $\rho_{a_k|i}^{B_k}$, the party operates a general quantum channel and obtains some information about the unknown measurements. The party conveys the information to another party, say $k'$. Based on this updated information, the $k'$-th party prepares a suitable state $\rho^{(k')}_{A_{k'}B_{k'}}$, applies the respective local measurement, and further executes LOCC distinguishing protocol for the candidates for $F_{\vec{a}|i}$.
This process continues sequentially, with each party reducing the set of possible measurements through local quantum operations and passing the updated information along. Importantly, the choice of the next party depends on both the classical outcome obtained by the previous party and the information received in the earlier rounds.  It is worth noting that for the special case of discriminating between two measurements (\textit{i.e.}, $v=2$), the adaptive strategy provides no advantage over the restricted one. The schematic representation of the circuits for the three types of protocols to distinguish bipartite measurements is provided in FIG. \ref{fig}.  

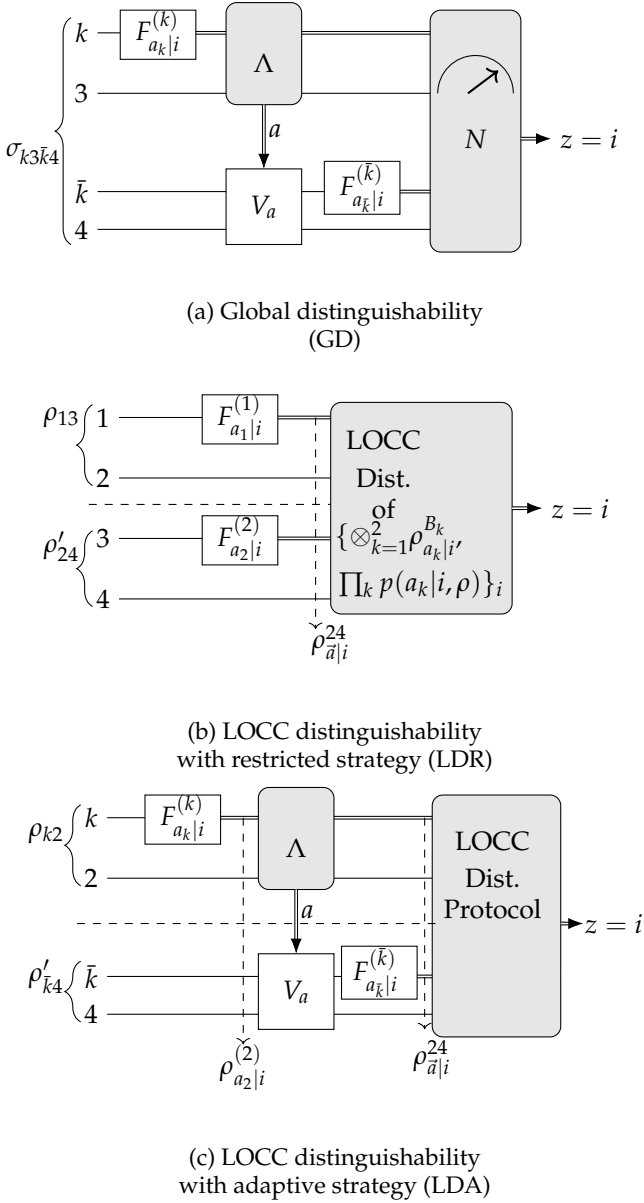
\begin{figure}
    \centering

    \subfloat[Global distinguishability (GD)]{
    \begin{minipage}{0.48\linewidth}
    \hspace*{-2.5cm}
    \centering
    \begin{tikzpicture}[scale=1,
      decoration={brace,mirror,amplitude=6pt}]

    \draw (0.5,0) -- (1.5,0);
    \draw (0.5,-0.8) -- (5.7,-0.8);
    \draw (0.5,-2.1) -- (3.7,-2.1);
    \draw (0.5,-2.6) -- (5.7,-2.6);
    \draw (2.7,0) -- (5.7,0);
    \draw[double] (1.6,0) -- (6,0);
    \draw[double] (4.5,-2.1) -- (5.7,-2.1);

    \node[left] at (0.5,0) {$k$};
    \node[left] at (0.5,-0.8) {$3$};
    \node[left] at (0.5,-2.1) {$\bar{k}$};
    \node[left] at (0.5,-2.6) {$4$};
    \node at (-0.35,-1.6) {$\sigma_{k3\bar{k}4}$};

    \draw[fill=white] (0.8,-0.35) rectangle (1.8,0.3);
    \node at (1.3,0) {$F_{a_k|i}^{(k)}$};

    \draw[decorate] (0.15,0.2) -- (0.15,-2.8);

    \draw[fill=gray!20,rounded corners]
        (2.2,-0.95) rectangle (3.2,0.4);
    \node at (2.7,-0.4) {$\Lambda$};



    \draw[double, -{Latex[length=2mm]}]
        (2.7,-0.95) -- (2.7,-1.8);
    \node at (2.85,-1.3) {$a$};

    \draw[fill=white] (2.2,-2.8) rectangle (3.2,-1.8);
    \node at (2.7,-2.3) {$V_a$};

    \draw[fill=white] (3.5,-2.4) rectangle (4.5,-1.7);
    \node at (4,-2.05) {$F_{a_{\bar{k}}|i}^{(\bar{k})}$};

     \draw[fill=gray!20,rounded corners]
        (4.9,-2.9) rectangle (6.1,0.3);
    \node at (5.5,-1.4) {$N$};

    \draw (5.0,-0.8) arc[start angle=180,end angle=0,radius=0.5 cm];

    \draw[->, thick] (5.4,-0.8) -- (5.8,-0.5);

    \draw[double, -{Latex[length=2mm]}]
        (6.1,-1.4) -- (6.5,-1.4);
    \node[right] at (6.5,-1.4) {$z=i$};

    \end{tikzpicture}
    \label{figb}
    \end{minipage}
    }

    \vspace{0.3cm}

    \subfloat[LOCC distinguishability with restricted strategy (LDR)]{
    \begin{minipage}{0.48\linewidth}
    \hspace*{-2cm}
    \centering
    \begin{tikzpicture}[scale=1,
      decoration={brace,mirror,amplitude=6pt}]

    \draw (0,0) -- (1.8,0);
    \draw (0,-0.8) -- (2.8,-0.8);
    \draw (0,-1.6) -- (1.8,-1.6);
    \draw (0,-2.4) -- (2.8,-2.4);
    \draw[double] (2,0) -- (2.8,0);
    \draw[double] (2,-1.6) -- (2.8,-1.6);

    \node[left] at (-0.4,0) {$\rho_{13}$};
    \node[left] at (-0.4,-1.7) {$\rho'_{24}$};

    \node[left] at (0,0) {$1$};
    \node[left] at (0,-0.8) {$2$};
    \node[left] at (0,-1.6) {$3$};
    \node[left] at (0,-2.4) {$4$};

    \draw[fill=white] (1.1,-0.35) rectangle (2.1,0.3);
    \node at (1.6,0) {$F_{a_1|i}^{(1)}$};

    \draw[decorate] (-0.35,0.17) -- (-0.35,-0.9);
    \draw[decorate] (-0.35,-1.5) -- (-0.35,-2.5);

    \draw[fill=white] (1.1,-1.95) rectangle (2.1,-1.3);
    \node at (1.6,-1.6) {$F_{a_2|i}^{(2)}$};

    \draw[dashed,->] (2.6,0) -- (2.6,-2.8);
    \draw[dashed,-] (-0.4,-1.15) -- (4,-1.15);

    \node at (2.8,-3) {$\rho_{\vec{a}|i}^{24}$};

    \draw[fill=gray!20,rounded corners]
        (2.8,-2.6) rectangle (5.2,0.2);

    \node at (3.5,-0.3) {LOCC};
    \node at (3.5,-0.8) {Dist.};
    \node at (3.5,-1.2) {of};
    \node at (3.7,-1.6) {$\{\otimes_{k=1}^2 \rho_{a_k|i}^{B_k}$,};
    \node at (4,-2.2) {$\prod_k p(a_k|i,\rho)\}_i$};

    \draw[double, -{Latex[length=2mm]}]
        (5.2,-1.2) -- (5.6,-1.2);
    \node[right] at (5.6,-1.2) {$z=i$};

    \end{tikzpicture}
    \label{figc}
    \end{minipage}
    }

    \hfill

    \subfloat[LOCC distinguishability with adaptive strategy (LDA)]{
    \begin{minipage}{0.48\linewidth}
    \hspace*{-2.2cm}
    \centering
    \begin{tikzpicture}[scale=1,
      decoration={brace,mirror,amplitude=6pt}]

    \draw (0,0) -- (4.7,0);
    \draw (0,-0.8) -- (4.7,-0.8);
    \draw (0,-2.1) -- (4.7,-2.1);
    \draw (0,-2.6) -- (4.7,-2.6);
    \draw[double] (1.3,0) -- (6,0);
    \draw[double] (4,-2.1) -- (6,-2.1);

    \node[left] at (0,0) {$k$};
    \node[left] at (0,-0.8) {$2$};
    \node[left] at (0,-2.1) {$\bar{k}$};
    \node[left] at (0,-2.6) {$4$};

    \node[left] at (-0.45,-0.2) {$\rho_{k2}$};
    \node[left] at (-0.45,-2.1) {$\rho'_{\bar{k}4}$};

    \draw[fill=white] (0.5,-0.35) rectangle (1.5,0.3);
    \node at (1,0) {$F_{a_k|i}^{(k)}$};

    \draw[decorate] (-0.35,0.17) -- (-0.35,-0.9);
    \draw[decorate] (-0.35,-1.9) -- (-0.35,-2.7);

     \draw[fill=gray!20,rounded corners]
        (2,-0.95) rectangle (3,0.4);
    \node at (2.5,-0.4) {$\Lambda$};



    \draw[double, -{Latex[length=2mm]}]
        (2.5,-0.95) -- (2.5,-1.8);
    \node at (2.65,-1.25) {$a$};

    \draw[fill=white] (2,-2.8) rectangle (3,-1.8);
    \node at (2.5,-2.3) {$V_a$};

    \draw[fill=white] (3.1,-2.4) rectangle (4.1,-1.7);
    \node at (3.55,-2.05) {$F_{a_{\bar{k}}|i}^{(\bar{k})}$};

    \draw[fill=gray!20,rounded corners]
        (4.3,-2.9) rectangle (6,0.3);

    \draw[dashed,->] (4.2,0) -- (4.2,-2.8);
    \draw[dashed,->] (1.8,0) -- (1.8,-3);
    \draw[dashed,-] (-0.4,-1.4) -- (4.4,-1.4);

    \node at (4.3,-3.2)
        {$\rho_{\vec{a}|i}^{24}$};
        \node at (1.8,-3.3)
        {$\rho_{a_2|i}^{(2)}$};

    \draw[double, -{Latex[length=2mm]}]
        (6,-1.4) -- (6.3,-1.4);

    \node at (5.1,-0.3) {LOCC};
    \node at (5.1,-0.8) {Dist.};
    \node at (5.1,-1.2) {Protocol};

    \node[right] at (6.2,-1.4) {$z=i$};

    \end{tikzpicture}
    \label{figd}
    \end{minipage}
    }

    \caption{Circuit diagrams illustrating different protocols for
    distinguishing a set of bipartite product measurements
    $F_{a_1|i}^{(1)}\otimes F_{a_2|i}^{(2)}$.
    Each measurement $F_{a_k|i}^{(k)}$ acts on subsystem $k$ of the
    input state. Single lines represent quantum systems, while
    double lines denote classical outcomes obtained from the
    measurements $F_{a_k|i}^{(k)}$ and $\Lambda$, which is a general channel. The outcome $z$ is the
    guess for the measurement index $i$. Except for the measurements
    $F_{a_k|i}^{(k)}$, which are given in advance, all other operations,
    including the input states, intermediate operations, and final
    measurements, can be freely chosen. The label
    ``LOCC Dist.\ of $\left\{
\bigotimes_{k=1}^{2} \rho^{B_k}_{a_k|i}\right\}_i,\{\prod_{k=1}^2 p(a_k|i,\rho^{(k)})\}_i $'' denotes protocol $\mathcal{P}_{\mathrm{LOCC}}\left[\left\{
\bigotimes_{k=1}^{2} \rho^{B_k}_{a_k|i}\right\}_i,\{\prod_{k=1}^2 p(a_k|i,\rho^{(k)})\}_i \right]$. The same thing is applied for the adaptive stretgy too. The variables
    $(k,\bar{k})$ can be either $(1,2)$ or $(2,1)$, corresponding to
    the two possible orderings of the local measurements in the
    adaptive strategy. The horizontal dashed line in panels
    (b) and (c) denotes the spatial separation between the two local
    parties.}
    \label{fig}
\end{figure}
Local restricted strategy can be read as an adaptive strategy in which all parties choose their probing state without implementing the information coming from the previous parties.
\begin{obs}
    For any set of multipartite measurements, the following holds:
    \bea\label{ob1}
    && LDR \leqslant LDA \leqslant GD
    \eea
\end{obs}
\begin{thm}\label{th1}
     For any set of bipartite qubit projective measurements, LDA is equivalent to LDR for perfect distinguishability.
\end{thm}
 LDA becomes strictly better than LDR when the second party needs different probing state to distinguish a different set of measurements. We prove that different set of distinguishable qubit projective measurements are always distinguishable with a common probe- a maximally entangled probe. Therefore, in the restricted strategy, the second party can always choose any maximally entangled state irrespective of the communication coming from the previous party.  
The details deferred to Appendix \ref{A}.

\section{Results}

For the first study, we directly dive into our main result, i.e., \emph{nonlocality without entanglement}. We provide an example of a set of tripartite measurements which are globally distinguishable but not locally distinguishable. Such a set of measurements acting on $\mathbb{C}^2\otimes\mathbb{C}^2$ are following:
\bea\label{nonlocality}
 G_{ab|1} &=& \sigma^z_a\otimes\sigma^z_b, \quad
  G_{ab|2} = \sigma^x_a\otimes\sigma^x_b,
\eea
where $\{\sigma^z_a\}_{a=1}^2 = \left\{\ket{0}\!\bra{0},\ket{1}\!\bra{1}\right\}$ and $
\{\sigma^x_a\}_{a=1}^2 = \{\ket{+}\!\bra{+},\ket{-}\!\bra{-}\}$.
\begin{thm}[Nonlocality without entanglement]\label{nl_thm}
    The measurements described by \eqref{nonlocality} are globally distinguishable but LOCC indistinguishable.
\end{thm}
\begin{proof}
    For the global protocol, the suitable probe is a product one, which is $\ket{\phi^-}\ket{\eta}\in\mathbbm{C}^4\otimes\mathbbm{C}^4$, where $\ket{\phi^-}=\frac{1}{\sqrt{2}}(\ket{00}-\ket{11})$. If first and fourth outcome occurs, the measurement is $ G_{ab|1}$ and clicking of second and third outcome implies the measurement is $ G_{ab|2}$.

    We now analyze LOCC protocols for distinguishing two measurements. Each party has access to two local measurements. Suppose Alice initiates the protocol. In a single round, Alice can eliminate at most one candidate; eliminating both would not provide any useful information to Bob. After choosing the optimum probe, Alice has to distinguish between two reduced states with associated probability for each outcome. Therefore, a two-outcome measurement suffices for this operation where each outcome of the measurement eliminates one possibility. This reduces the problem to distinguishing the pair of reduced states for the two outcomes of $\sigma^z_a$ and $\sigma^x_a$. Therefore, the problem is nothing but the discrimination of the measurements $\sigma^z_a$ and $\sigma^x_a$. From the proof of Theorem \ref{th1}, we know that only maximally entangled state can act as the probe for perfect discrimination of two qubit projective measurements. For $\sigma^z_a$ and $\sigma^x_a$, there does not exists $\mathcal{D}\big[\{\sigma_a^z,\sigma_a^x\},\{1/2,1/2\}\big]$
 with maximally entangled probe as eventually, there does not exist $\mathcal{P}[\{\ket{0},\ket{+}\},\{1/2,1/2\}]$ and $\mathcal{P}[\{\ket{1},\ket{-}\},\{1/2,1/2\}]$.
    Therefore, Alice cannot perfectly distinguish $\sigma^z_a$ and $\sigma^x_a$. Consequently, Bob must eventually distinguish between the two measurements of his side. He also possess the same two measurements as Alice. henceforth, a similar argument shows that Bob cannot distinguish two measurements. Now we check if any multi-round LOCC protocol can distinguish this set. Suppose, there exists a finite-round LOCC
protocol that perfectly distinguishes $G_{ab|1}$ and $G_{ab|2}$. The entire
post-processing stage, including all local quantum operations, local
ancillas, intermediate measurements, classical messages, and adaptive
choices, defines a completely positive trace-preserving map
$\Lambda_{\mathrm{LOCC}}$. Its final output may be taken to be a
classical register containing the protocol's guess.
We use the fidelity
$
    F(\rho,\sigma)
    :=
    \left\|
        \sqrt{\rho}\sqrt{\sigma}
    \right\|_1^2.
$
The fidelity is multiplicative under tensor products,
$F(\rho_1\otimes\rho_2,\sigma_1\otimes\sigma_2)
    =
    F(\rho_1,\sigma_1)F(\rho_2,\sigma_2)$,
and is monotone under every quantum channel:
$
    F\!\left(
        \Lambda(\rho),\Lambda(\sigma)
    \right)
    \geq
    F(\rho,\sigma).$
 Perfect discrimination requires
the final outputs associated with $G_{ab|1}$ and $G_{ab|2}$ to have disjoint
supports. Hence,
 $   F\!\left(
        \Lambda_{\mathrm{LOCC}}(\omega_1),
        \Lambda_{\mathrm{LOCC}}(\omega_2)
    \right)
    =0,
$ where $\omega_1=\alpha_1\otimes\beta_1,\omega_2=\alpha_2\otimes\beta_2$ are the joint reduced states after the measurements are implemented on the probe.
By monotonicity of the fidelity,
\begin{equation}
\begin{split}
    0
    &=
    F\!\left(
        \Lambda_{\mathrm{LOCC}}(\omega_1),
        \Lambda_{\mathrm{LOCC}}(\omega_2)
    \right)                                                  \geq
    F(\omega_1,\omega_2)
    \geq 0.
\end{split}
\end{equation}
It follows that
$
    F(\omega_1,\omega_2)=0.
$
Consequently, we obtain
$
    F(\omega_1,\omega_2)
    =
    F(\alpha_1\otimes\beta_1,
      \alpha_1\otimes\beta_2)                                 =
    F(\alpha_1,\alpha_2)F(\beta_1,\beta_2)                    =0.
$
At least, we need $ F(\alpha_1,\alpha_2)=0$ or $F(\beta_1,\beta_2)$, which means the reduced states coming from $\sigma^z_a$ and $\sigma_a^x$ needs to be orthogonal, which is not possible.
    
If Bob instead initiates the protocol, the same limitation applies. Since $\sigma^z_b$ and $\sigma^x_b$ are indistinguishable, he cannot eliminate any measurement perfectly in the first step. Consequently, regardless of which party initiates the protocol, the set of unitaries is not perfectly distinguishable under either LDR or LDA strategies. 
\end{proof}

Now our next venture is to compare the two types of local protocols. The relation \ref{ob1} suggests that adaptive strategy can be better than the restricted strategy. But Theorem \ref{th1} also tells that both the strategies are equivalent in the case of qubit measurements as the perfect discrimination concerns.
Let us consider following bipartite projective measurements acting on $\mathbb{C}^3\otimes \mathbb{C}^3$.
\bea\label{adapt}
L_{ab|1} &=& \Omega^z_a\otimes \Omega^z_b,
L_{ab|2} =  \Omega^z_a\otimes \Omega^x_b,\nonumber\\
L_{ab|3} &=&  \Omega^{'z}_a\otimes \Omega^z_b,
L_{ab|4} =  \Omega^{'z}_a\otimes \Omega^{'x}_b,\nonumber\\
L_{ab|5} &=&  \Tilde{\Omega}^z_a\otimes \Omega^z_b,
L_{ab|6} =  \Tilde{\Omega}^z_a\otimes \Tilde{\Omega}^x_b,
\eea
where 
\bea
\{\Omega^z_a\}_{a=1}^3&=&\left\{\ket{0}\!\bra{0},\ket{1}\!\bra{1}, \ket{2}\!\bra{2}\right\},\nonumber\\
\{\Omega^{'z}_a\}_{a=1}^3 &=&\{\ket{2}\!\bra{2},\ket{0}\!\bra{0},\ket{1}\!\bra{1}\},\nonumber\\
\{\Tilde{\Omega}^z_a\}_{a=1}^3 &=&\{\ket{1}\!\bra{1},\ket{2}\!\bra{2},\ket{0}\!\bra{0}\},\nonumber\\  \{\Omega^x_a\}_{a=1}^3 &=&\{\frac12\ket{1-2}\!\bra{1-2},\ket{0}\!\bra{0}, \frac12\ket{1+2}\!\bra{1+2}\},\nonumber\\
\{\Omega^{'x}_a\}_{a=1}^3 &=&\{\frac12\ket{0+2}\!\bra{0+2},\frac12\ket{0-2}\!\bra{0-2}, \ket{1}\!\bra{1}\},\nonumber\\
\{\Tilde{\Omega}^x_a\}_{a=1}^3 &=&\{\ket{2}\!\bra{2},\frac12\ket{0+1}\!\bra{0+1}, \frac12\ket{0-1}\!\bra{0-1}\}.\nonumber\\
\eea
\begin{thm}[Hierarchy between local protocols]\label{th3}
The set of measurements defined at \eqref{adapt} are distinguishable in one-way LDA but indistinguishable in any finite-round LDR.
\end{thm}
 For LDA, Alice chooses the maximally entangled probe and Bob chooses different product probes upon communication from Alice. We proved that there is no common probe by which Bob can distinguish the necessary pairs of measurements. With fidelity monotonicity, we proved the further rounds are useless. We also prove the starter of the protocol is irrelevant to the success of the task. The proof is deferred to Appendix \ref{C}.

We move into another example of four measurements acting on $\mathbb{C}^2\otimes\mathbb{C}^2$. This set of measurements shows the asymmetry of the result in distinguishability with respect to the starter of the protocol. Note that this was not the case for the last result. Let us take the following four measurements:
\bea\label{asym}
H_{ab|1} &=& \sigma^z_a\otimes\sigma^z_b, \ \ 
H_{ab|2} = \overline{\sigma}^z_a\otimes\sigma^x_b,\nonumber\\
H_{ab|3} &=& \sigma^z_a\otimes\overline{\sigma}^z_b, \ \ 
H_{ab|4} = \overline{\sigma}^z_a\otimes\overline{\sigma}^x_b,
\eea
$\overline{\sigma}^z$ and $\overline{\sigma}^x$ are $\sigma^z$ and $\sigma^x$ respectively with reverse outcome order.
\begin{thm}[Asymmetry of starter]\label{th_asym}
    The measurements described by \eqref{asym} are locally distinguishable if Alice starts. If Bob starts, the measurements remain locally indistinguishable.
\end{thm}
The proof is deffered to Appendix \ref{B}.

Next, we examine two multipartite qubit product measurements that are perfectly distinguishable by a global strategy but are not perfectly
distinguishable by any finite-round LOCC protocol across any
bipartition, even when the two parties lying on the same side of a bipartition are allowed arbitrary joint quantum operations. This phenomenon is commonly known as 'genuine nonlocality without entanglement.' 

Consider the two $w$-partite qubit product measurements
\bea\label{G0}
& \mathcal{G}_{a_1a_2\cdots a_w|1}
=
\bigotimes_{i=1}^w\sigma^z_{a_i}, \ \ 
\mathcal{G}_{a_1a_2\cdots a_w|2}
=
\bigotimes_{i=1}^w\sigma^u_{a_i} ,
\eea 
where $\sigma^z_{a_i}=\sigma^z_a,\{\sigma^u_{a_i}\}_{a_i}=\left\{
|u\rangle\!\langle u|,
|u^\perp\rangle\!\langle u^\perp| \right\},$
with $|u\rangle =
\cos(\pi/2w)|0\rangle
+ \sin(\pi/2w)|1\rangle.$

\begin{thm}[Strong nonlocality without entanglement]
\label{thm:bipartition-wise-indistinguishability}
The $w$-partite qubit measurement pair $\mathcal{G}_{a_1a_2\cdots a_w|1}$ and $\mathcal{G}_{a_1a_2\cdots a_w|2}$, given by \eqref{G0}, are perfectly
distinguishable globally, whereas they are not perfectly
distinguishable by finite-round LOCC across any of the bipartitions.
\end{thm}
The proof is provided in the Appendix \ref{D}. Note that Theorem~\ref{nl_thm} is a particular instance of this theorem.

\section{Application to quantum data hiding}

Quantum data hiding encodes classical information so that it can be recovered by global quantum operations but not by spatially separated parties restricted to LOCC. Conventionally, the information is encoded in quantum states, and data hiding has been closely connected to NLWE  \cite{DiVincenzo_ieee}. Here we point out an alternative paradigm based on \emph{quantum-process-based nonlocality without entanglement}, where the hidden information is encoded directly in the identity of a distributed quantum process rather than in a quantum state.

Our results provide a particularly elementary realization using quantum measurements. A single classical bit can be encoded into the choice between the two bipartite qubit product measurements introduced above: the bit can be recovered perfectly when the measurement devices are accessed jointly, whereas perfect recovery is impossible under LOCC. Taking polarization as the qubit degree of freedom, these measurements can be implemented simply using polarizing beam splitters (PBSs) in the $Z$ and $X$ polarization bases. Thus, the information is effectively hidden in the configuration of robust macroscopic elementary optical elements, making the scheme experimentally straightforward.

A complementary realization is possible with unitary processes. Recent results show that there exist pairs of bipartite qubit product unitaries that are perfectly distinguishable globally but not by LOCC \cite{Manna_pra_2026}. Hence, one classical bit can likewise be hidden in the choice of a distributed unitary transformation. For polarization qubits, the relevant single-qubit unitaries can be implemented using elementary wave-plate transformations, such as quarter-wave or half-wave plates. Together, these observations suggest quantum-process-based NLWE as a simple resource for quantum data hiding.

\section{Discussion}
In summary, we have shown that extending nonlocality without entanglement to quantum measurements uncovers qualitatively new phenomena and reveals intrinsic limitations of measurement discrimination under global and LOCC operations. Our results establish several unexpected separations, demonstrating that the different discrimination paradigms are neither equivalent nor universally ordered. The most interesting feature is genuinely nonlocal measurements: one can construct two distinguishable multipartite measurements that remain locally indistinguishable across every bipartition. This is in stark contrast to the analogous scenario for quantum states, where such a phenomenon cannot arise for two multipartite product states. To the best of our knowledge, an analogous phenomenon has not yet been demonstrated for multipartite product-state discrimination, for any number of states.

Several directions remain open. A natural next step is to characterize the optimal success probability for discriminating arbitrary measurement sets, yielding a quantitative characterization of the gap between discrimination paradigms, analogous to that for quantum states. Here we restrict our analysis to demolition measurements, where the post-measurement quantum system is inaccessible. A natural extension is to quantum instruments, for which the post-measurement state provides an additional resource for discrimination \cite{manna_25,eid2026}, potentially revealing new manifestations of nonlocality without entanglement. Our analysis has focused on the single-shot regime. A direction for future work is to explore multishot nonlocality without entanglement, where multiple uses of the unknown measurement are available, and to determine whether local indistinguishability persists or can be activated by collective probing across multiple uses. More broadly, the framework developed here may serve as a foundation for new protocols and applications in quantum information processing that exploit the contrast between global and local measurement discrimination.

\subsection*{Acknowledgment}
DS acknowledges support from STARS (STARS/STARS-2/2023-0809), Govt. of India. The use of generative artificial intelligence tools was limited to language editing.

\bibliography{ref}
\appendix
\onecolumngrid
\section{Proof of Theorem 1}\label{A}
Before going to the proof of the theorem, we present a lemma which will be useful for the proof.
\begin{lemma}[Necessary condition for perfect discrimination]
\label{lem:general-unitary-condition}
Let
$\mathcal{M}_0 = \left\{
|x\rangle\langle x|
\right\}_{x=0}^{d-1}$
and
$\mathcal{M}_1 = \left\{
U|x\rangle\langle x|U^\dagger
\right\}_{x=0}^{d-1}$
be two rank-one projective measurements on $\mathbb{C}^d$, where
$U$ is an arbitrary unitary operator. If $\mathcal{M}_0$ and
$\mathcal{M}_1$ can be perfectly distinguished with a single use,
allowing an arbitrary ancillary system, then there exists a density
operator $\rho$ on the measured subsystem such that
\begin{equation}
    \langle x|\rho U|x\rangle=0
\qquad
\text{for every }x=0,\ldots,d-1.
\end{equation}
Consequently, a necessary condition for perfect discrimination is
\begin{equation}\label{rU=0}
    \operatorname{Tr}(\rho U)=0.
\end{equation}
\end{lemma}
\begin{proof}
Let $|\psi\rangle_{SR}$ be an arbitrary pure probe, where $S$ is the
$d$-dimensional system on which the unknown measurement acts and $R$ is an arbitrary ancillary system. It is sufficient to consider pure
probes, since every mixed probe can be purified by enlarging the
ancillary system. Let
$\rho = \operatorname{Tr}_{R}
\left(|\psi\rangle\langle\psi|
\right)$ be the reduced state on the measured system $S$.

Suppose that the classical outcome of the unknown measurement is
$x$. If the measurement is $\mathcal{M}_0$, the corresponding
unnormalized conditional state vector on the ancillary system is
\begin{equation}
|\alpha_x\rangle_R
=
(\langle x|\otimes I_R)|\psi\rangle_{SR}.
\end{equation}
If the measurement is $\mathcal{M}_1$, the corresponding conditional
state vector is
\begin{equation}
|\beta_x\rangle_R
=
(\langle x|U^\dagger\otimes I_R)|\psi\rangle_{SR}.
\end{equation}
For perfect discrimination, whenever the same classical outcome $x$
can arise from both measurements, the corresponding conditional
states must be orthogonal. Therefore,
$\langle\beta_x|\alpha_x\rangle=0,\ \forall\,x.$
Using the definition of $\rho$, we obtain
\begin{equation}
\langle\beta_x|\alpha_x\rangle
=
\langle\psi|
\left(
U|x\rangle\langle x|\otimes I_R
\right)
|\psi\rangle
=
\langle x|\rho U|x\rangle.
\end{equation}
Hence perfect discrimination requires $\langle x|\rho U|x\rangle=0, \forall\,x.$
Summing over the complete computational basis gives \eqref{rU=0}. This completes the proof.
\end{proof}

\emph{Proof of Theorem \ref{th1}: }
One-way LDA becomes strictly better than LDR when the last party needs different probing state to distinguish a different set of measurements. Let us take two qubit projective measurements such that $F_{1|1}=\ket{0}\bra{0}$ and $F_{1|2}=\ket{\eta}\bra{\eta}$.

Let the reduced density operator of the probe on the measured qubit be
\begin{equation}
    \rho=
    \begin{pmatrix}
        r & s\\
        s^{*} & t
    \end{pmatrix},
\end{equation}
where $\rho>0, \tr(\rho)=1,$ corresponds to a full-Schmidt-rank entangled probe. Let
$|\eta\rangle
    = \alpha|0\rangle+\beta|1\rangle,
    \
    |\eta^{\perp}\rangle
    =
    -\beta^{*}|0\rangle+\alpha^{*}|1\rangle,$
with $|\alpha|^{2}+|\beta|^{2}=1$.
For perfect discrimination, the conditional states corresponding to each common outcome must be mutually orthogonal. Using Lemma~\ref{lem:general-unitary-condition}, we know that for perfect distinguishability, 
$\langle 0|\rho|\eta\rangle=0,
    \langle 1|\rho|\eta^{\perp}\rangle=0.$
Using the above parametrization, we get
\begin{align}
    r\alpha+s\beta &=0,
    \label{eq:first_R_condition}\\
    -s^{*}\beta^{*}+t\alpha^{*}&=0.
    \label{eq:second_R_condition}
\end{align}
Suppose that $\alpha\neq0$ and $\beta\neq0$. From
Eq.~\eqref{eq:first_R_condition}, we obtain $s=-\frac{r\alpha}{\beta},$ and hence
$s^{*}=-\frac{r\alpha^{*}}{\beta^{*}},$ where we have used the fact that $r$ is real. Substituting this into
Eq.~\eqref{eq:second_R_condition} yields $r\alpha^{*}+t\alpha^{*}=0.$
Since $\alpha\neq0$, this implies
$r+t=0,$ which is incompatible with $\tr(\rho)=1$.

Thus, for a full-Schmidt-rank probe one must have $\alpha=0,$ which implies $|\eta\rangle=e^{i\theta}|1\rangle$ for some irrelevant phase $\theta$. Therefore, a full-Schmidt-rank probe can perfectly discriminate the
two binary qubit projective measurements only when the corresponding
projectors associated with the same outcome are orthogonal. In this case, a maximally entangled probe or product system also achieves perfect discrimination. Theorem $5$ of ref. \cite{manna_25} tells that all the maximally entangled states are equivalent for
evaluation of distinguishability of any two general qubit rank-one measurements. So any two distinguishable qubit projective measurements are always distinguishable with any maximally entangled probe.

\section{Proof of Theorem \ref{th3}}\label{C}
For LDA, Alice starts the protocol. She chooses the maximally entangled state as the probe and measure the reduced state in $\{\ket{0},\ket{1},\ket{2}\}$ basis. If first outcome occurs, Alice knows it is either first or second measurement. Upon communication, Bob chooses the probe $\ket{1}\ket{s}$. If second outcome occurs, Bob knows it is the first measurement and for other two outcomes, he knows it is second measurement.

If Alice got second outcome, Bob chooses the probe $\ket{2}\ket{s}$ and distinguish the measurements $L_{ab|3}$ and $L_{ab|4}$. For the third outcome of Alice, Bob will select $\ket{0}\ket{s}$ as the probe and successfully discriminate the fifth and sixth measurement. $\ket{s}$ is arbitrary.


Let Alice and Bob prepare arbitrary local probe states
$\rho_{AA'}$ and $\rho_{BB'}$, respectively, where $A'$ and $B'$
denote arbitrary local ancillary systems. In an LDR protocol, these
probe states are chosen independently and before any classical
communication takes place. In particular, the joint probe is of the
product form
\begin{equation}
    \rho_{AA'}\otimes\rho_{BB'}.
\end{equation}

After the unknown bipartite measurement $L_i$ is applied, let
$\alpha_i$ and $\beta_i$ denote the resulting local
classical--quantum states held by Alice and Bob, respectively. The
classical registers containing the local measurement outcomes are
included in $\alpha_i$ and $\beta_i$. Since every $L_{ab|i}$ is a product
measurement and the initial probe is a product across the
Alice--Bob partition, the joint state prior to the subsequent LOCC
protocol has the form
\begin{equation}
    \omega_i=\alpha_i\otimes\beta_i.
    \label{eq:ldr-product-output}
\end{equation}

The six measurements are grouped into the three pairs
\begin{equation}
    (L_{ab|1},L_{ab|2}),\qquad
    (L_{ab|3},L_{ab|4}),\qquad
    (L_{ab|5},L_{ab|6}).
\end{equation}
Within each pair, Alice's local measurement is identical. Consequently,
her complete local classical--quantum state is also identical for the
two measurements in each pair:
\begin{equation}
    \alpha_1=\alpha_2,\qquad
    \alpha_3=\alpha_4,\qquad
    \alpha_5=\alpha_6.
    \label{eq:alice-identical-states}
\end{equation}

Suppose, toward a contradiction, that there exists a finite-round LOCC
protocol that perfectly distinguishes the six measurements. The entire
post-processing stage, including all local quantum operations, local
ancillas, intermediate measurements, classical messages, and adaptive
choices, defines a completely positive trace-preserving map
$\Lambda_{\mathrm{LOCC}}$. Its final output may be taken to be a
classical register containing the protocol's guess.


Consider first the pair $(L_{ab|1},L_{ab|2})$. Perfect discrimination requires
the final outputs associated with $L_{ab|1}$ and $L_{ab|2}$ to have disjoint
supports. Hence,
\begin{equation}
    F\!\left(
        \Lambda_{\mathrm{LOCC}}(\omega_1),
        \Lambda_{\mathrm{LOCC}}(\omega_2)
    \right)
    =0,
    \label{eq:final-fidelity-zero-12}
\end{equation}
where $F$ is the fidelity, which is defined in the proof of Theorem \ref{nl_thm}. 
By monotonicity of the fidelity,
we can write
\begin{equation}
    F(\omega_1,\omega_2)=0,
    \label{eq:input-fidelity-zero-12}
\end{equation}
which turns into
\begin{equation}
\begin{split}
    F(\omega_1,\omega_2)
    &=
    F(\alpha_1\otimes\beta_1,
      \alpha_1\otimes\beta_2)                                 =
    F(\alpha_1,\alpha_1)F(\beta_1,\beta_2)                    =
    F(\beta_1,\beta_2).
\end{split}
\end{equation}
Therefore,
\begin{equation}
    F(\beta_1,\beta_2)=0.
    \label{eq:bob-orthogonal-12}
\end{equation}
Thus Bob's local classical--quantum states corresponding to $L_{ab|1}$ and
$L_{ab|2}$ must already have orthogonal supports. Equivalently, Bob's fixed
local probe must perfectly distinguish the two local measurements
appearing in the pair $(L_{ab|1},L_{ab|2})$.

The same argument applies to the other two pairs. Since
$\alpha_3=\alpha_4$ and $\alpha_5=\alpha_6$, perfect discrimination
also implies
\begin{equation}
    F(\beta_3,\beta_4)=0,
    \qquad
    F(\beta_5,\beta_6)=0.
    \label{eq:bob-orthogonal-all-pairs}
\end{equation}
Consequently, the same fixed Bob probe must simultaneously perfectly
distinguish the three local measurement pairs associated with
\begin{equation}
    (L_{ab|1},L_{ab|2}),\qquad
    (L_{ab|3},L_{ab|4}),\qquad
    (L_{ab|5},L_{ab|6}).
    \label{eq:three-required-pairs}
\end{equation}
However, we will show next that no single Bob probe state can
simultaneously perfectly distinguish all three pairs in
Eq.~\eqref{eq:three-required-pairs}. This contradicts
Eqs.~\eqref{eq:bob-orthogonal-12} and
\eqref{eq:bob-orthogonal-all-pairs}. Therefore, no finite-round LDR
protocol can perfectly distinguish the six measurements.

The argument does not depend on which party performs the first LOCC
operation. If Alice starts, all of her operations and messages are
included in the channel $\Lambda_{\mathrm{LOCC}}$, and fidelity
monotonicity shows that her identical local states cannot make
nonorthogonal Bob states perfectly distinguishable. If Bob starts, the
same reasoning applies: the full Bob-first adaptive protocol is again
a quantum channel $\Lambda_{\mathrm{LOCC}}$, while Alice's input state
remains identical within each of the three pairs. Hence Bob's fixed
probe must again distinguish all three local pairs simultaneously.

Now we will prove that no single Bob probe state can
simultaneously perfectly distinguish all three pairs if Alice begins. Alice's best strategy is the previous one. In restricted strategy, Bob choose the probe without implementing the information of Alice. The most general probe would be a qutrit-qutrit entangled state $\ket{\phi_p}=\sum_{g,h=0}^2\gamma_{gh}\ket{gh}$.
Firstly we are considering the perfect distinguishability of first two measurements of Bob which are $\{\Omega^z_a\}_a$ and $\{\Omega^x_a\}_a$. For perfect distinguishability, 
$\mathcal{D}\left[\left\{\Omega_a^z,\Omega_a^x\right\}_a,\left\{1/2,1/2\right\}\right]$ must exist. Consequently, the following quantities exist.  $\mathcal{P}\left[\left\{ \frac{\gamma_{00}\ket{0}+\gamma_{01}\ket{1}+\gamma_{02}\ket{2}}{\sqrt{p_1}},\frac{(\gamma_{10}-\gamma_{20}) \ket{0}+(\gamma_{11}-\gamma_{21})\ket{1}+(\gamma_{12}-\gamma_{22})\ket{2}}{\sqrt{p_2}}\right\},\left\{p_1,p_2/2\right\}\right]$, $ \mathcal{P} \left[\left \{ \frac{\gamma_{10}\ket{0}+\gamma_{11}\ket{1}+\gamma_{12}\ket{2}}{\sqrt{p_3}},\frac{\gamma_{00}\ket{0}+\gamma_{01}\ket{1}+\gamma_{02}\ket{2}}{\sqrt{p_4}}\right\},\left\{p_3,p_4\right\}\right]$ and $   \mathcal{P}\Big[\left\{ \frac{\gamma_{20}\ket{0}+\gamma_{21}\ket{1}+\gamma_{22}\ket{2}}{\sqrt{p_5}},\frac{(\gamma_{10}+\gamma_{20}) \ket{0}+(\gamma_{11}+\gamma_{21})\ket{1}+(\gamma_{12}+\gamma_{22})\ket{2}}{\sqrt{p_6}}\right\},$ $\left\{p_5,p_6/2\right\}\Big],$
where $p_1=p_4=|\gamma_{00}|^2+|\gamma_{01}|^2+|\gamma_{02}|^2, p_2 = |\gamma_{10}-\gamma_{20}|^2+|\gamma_{11}-\gamma_{21}|^2+|\gamma_{12}-\gamma_{22}|^2, p_3=|\gamma_{10}|^2+|\gamma_{11}|^2+|\gamma_{12}|^2, p_5=|\gamma_{20}|^2+|\gamma_{21}|^2+|\gamma_{22}|^2$ and $p_6=|\gamma_{10}+\gamma_{20}|^2+|\gamma_{11}+\gamma_{21}|^2+|\gamma_{12}+\gamma_{22}|^2$.
For $\mathcal{D}\big[\{\Omega^z_a,\Omega_a^x\}_{a,x},\{1/2,1/2\}\big]$
, all three sets of reduced states should be distinguishable, i.e., orthogonal. Using this fact, we get the following three equations.
\bea
&& \gamma_{00}^*(\gamma_{10}-\gamma_{20})+\gamma_{01}^*(\gamma_{11}-\gamma_{21})+\gamma_{02}^*(\gamma_{12}-\gamma_{22})=0;\label{D3}\\
&& \gamma_{00}^*\gamma_{10}+\gamma_{01}^*\gamma_{11}+\gamma_{02}^*\gamma_{12}=0;\label{D4}\\
&& \gamma_{20}^*(\gamma_{10}+\gamma_{20})+\gamma_{21}^*(\gamma_{11}+\gamma_{21})+\gamma_{22}^*(\gamma_{12}+\gamma_{22})=0;\label{D5}
\eea
From \eqref{D3} and \eqref{D4}, we arrive at,
\be\label{gamma1}
\gamma_{00}^*\gamma_{20}+\gamma_{01}^*\gamma_{21}+\gamma_{02}^*\gamma_{22}=0;
\ee
 As this is the restricted strategy, $\ket{\phi_p}$ should perfectly discriminate the third and fourth measurement of Bob, i.e., $\{\Omega^z_a\}_a$ and $\{\Omega_a^{'x}\}$. Similarly, we check
$\mathcal{D}_{\ket{\phi_p}}\left[\left\{\Omega_a^z,\Omega_a^{'x}\right\}_a,\left\{1/2,1/2\right\}\right]$ with probe $\ket{\phi_p}$. That means $\mathcal{P}\left[\left\{ \frac{\gamma_{00}\ket{0}+\gamma_{01}\ket{1}+\gamma_{02}\ket{2}}{\sqrt{p_1}},\frac{(\gamma_{00}+\gamma_{20}) \ket{0}+(\gamma_{01}+\gamma_{21})\ket{1}+(\gamma_{02}+\gamma_{22})\ket{2}}{\sqrt{l_1}}\right\},\left\{p_1,l_1/2\right\}\right]$, $ \mathcal{P}\left[\left\{ \frac{\gamma_{10}\ket{0}+\gamma_{11}\ket{1}+\gamma_{12}\ket{2}}{\sqrt{p_3}},\frac{(\gamma_{00}-\gamma_{20}) \ket{0}+(\gamma_{01}-\gamma_{21})\ket{1}+(\gamma_{02}-\gamma_{22})\ket{2}}{\sqrt{l_2}}\right\},\left\{p_3,l_2/2\right\}\right] $ and $  \mathcal{P}\Big[\left\{ \frac{\gamma_{20}\ket{0}+\gamma_{21}\ket{1}+\gamma_{22}\ket{2}}{\sqrt{p_5}},\frac{\gamma_{10}\ket{0}+\gamma_{11}\ket{1}+\gamma_{12}\ket{2}}{\sqrt{p_3}}\right\},$ $\left\{p_5,p_3\right\}\Big]  \Big)$ exist,
where $l_1=|\gamma_{00}+\gamma_{20}|^2+|\gamma_{01}+\gamma_{21}|^2+|\gamma_{02}+\gamma_{22}|^2$ and $l_2=|\gamma_{00}-\gamma_{20}|^2+|\gamma_{01}-\gamma_{21}|^2+|\gamma_{02}-\gamma_{22}|^2$.
From the orthogonality condition of the set of reduced states, we get 
\bea
&& \gamma_{00}^*(\gamma_{00}+\gamma_{20})+\gamma_{01}^*(\gamma_{01}+\gamma_{21})+\gamma_{02}^*(\gamma_{02}+\gamma_{22})=0;\label{D8}\\
&& \gamma_{10}^*(\gamma_{00}-\gamma_{20})+\gamma_{11}^*(\gamma_{01}-\gamma_{21})+\gamma_{12}^*(\gamma_{02}-\gamma_{22})=0;\label{D9}\\
&&\gamma_{20}^*\gamma_{10}+\gamma_{21}^*\gamma_{11}+\gamma_{22}^*\gamma_{12}=0;\label{D10}
\eea
Introducing \eqref{gamma1} into \eqref{D8}, we find that
\be
|\gamma_{00}|^2+|\gamma_{01}|^2+|\gamma_{02}|^2=0,
\ee
which implies 
\be\label{De12}
\gamma_{00}=\gamma_{01}=\gamma_{02}=0.
\ee
Now for the last two measurements, we want to check the existence of
 $\mathcal{D}_{\ket{\phi_p}}\left[\left\{\Omega_a^z,\tilde{\Omega}_a^{x}\right\}_a,\left\{1/2,1/2\right\}\right]$, which implies the existence of
 $\mathcal{P}\left[\left\{ \frac{\gamma_{00}\ket{0}+\gamma_{01}\ket{1}+\gamma_{02}\ket{2}}{\sqrt{p_1}},\frac{\gamma_{20} \ket{0}+\gamma_{21}\ket{1}+\gamma_{22}\ket{2}}{\sqrt{l_1}}\right\},\left\{p_1,l_1\right\}\right],$ $ \mathcal{P}\left[\left\{ \frac{\gamma_{10}\ket{0}+\gamma_{11}\ket{1}+\gamma_{12}\ket{2}}{\sqrt{p_3}},\frac{(\gamma_{00}+\gamma_{10}) \ket{0}+(\gamma_{01}+\gamma_{11})\ket{1}+(\gamma_{02}+\gamma_{12})\ket{2}}{\sqrt{k_1}}\right\},\left\{p_3,k_1/2\right\}\right] $ and \newline$ \mathcal{P}\Big[\Big\{ \frac{\gamma_{20}\ket{0} +\gamma_{21}\ket{1}+ \gamma_{22}\ket{2}}{\sqrt{p_5}}$,
 $\frac{(\gamma_{00}-\gamma_{10}) \ket{0}+ (\gamma_{01}-\gamma_{11})\ket{1}+ (\gamma_{02}-\gamma_{12})\ket{2}}{\sqrt{k_2}}\Big\},\left\{p_5,k_2/2\right\}\Big],$
where $k_1=|\gamma_{00}+\gamma_{10}|^2+|\gamma_{01}+\gamma_{11}|^2+|\gamma_{02}+\gamma_{12}|^2$ and $k_2=|\gamma_{00}-\gamma_{10}|^2+|\gamma_{01}-\gamma_{11}|^2+|\gamma_{02}-\gamma_{12}|^2$.
From the orthogonality condition of the set of states, we get 
\bea
&&\gamma_{20}^*\gamma_{00}+\gamma_{21}^*\gamma_{01}+\gamma_{22}^*\gamma_{02}=0;\label{D13}\\
&& \gamma_{10}^*(\gamma_{00}+\gamma_{10})+\gamma_{11}^*(\gamma_{01}+\gamma_{11})+\gamma_{12}^*(\gamma_{02}+\gamma_{12})=0;\label{D14}\\
&& \gamma_{20}^*(\gamma_{00}-\gamma_{10})+\gamma_{21}^*(\gamma_{01}-\gamma_{11})+\gamma_{22}^*(\gamma_{02}-\gamma_{12})=0;\label{D15}
\eea
From \eqref{D14} and \eqref{De12}, we get,
\be\label{De17}
\gamma_{10}=\gamma_{11}=\gamma_{12}=0.
\ee
Implementing \eqref{De17} into \eqref{D5}, we finally get
\be\label{De16}
\gamma_{20}=\gamma_{21}=\gamma_{22}=0.
\ee
From \eqref{De12}, \eqref{De17} and \eqref{De16}, we can infer that there does not exist any probe which can distinguish all three sets simultaneously.

\section{Proof of Theorem 4}\label{B} 
\begin{lemma}
Let
$|\psi\rangle_{AB}=\sum_{i=1}^{d}s_i\,|\chi_i\rangle_A\otimes|\zeta_i\rangle_B$
be the Schmidt decomposition of a bipartite pure state, where
\(s_i\ge 0\) and \(\sum_i s_i^2=1\). Define
$
\sigma(H)=\operatorname{Tr}_A\!\left[(H\otimes I)\,|\psi\rangle\langle\psi|\right],
$
where \(H\) is an operator acting on \(\mathcal{H}_A\). Then
$
\sigma(H)=S H^{T} S,
$
where
$
S=\operatorname{diag}(s_1,\ldots,s_d),
$
and \(H^T\) denotes the transpose of \(H\) in the Schmidt basis
\(\{|\chi_i\rangle\}\).
\end{lemma}

\begin{proof}
Using the Schmidt decomposition,
$
|\psi\rangle\langle\psi|
=
\sum_{i,j=1}^{d}
s_i s_j
\,|\chi_i\rangle\langle\chi_j|
\otimes
|\zeta_i\rangle\langle\zeta_j|.
$
Hence,
$
(H\otimes I)|\psi\rangle\langle\psi|
=
\sum_{i,j=1}^{d}
s_i s_j
\,H|\chi_i\rangle\langle\chi_j|
\otimes
|\zeta_i\rangle\langle\zeta_j|.
$

Taking the partial trace over subsystem \(A\),
\begin{align*}
\sigma(H)
&=
\operatorname{Tr}_A\!\left[(H\otimes I)|\psi\rangle\langle\psi|\right] 
=
\sum_{i,j=1}^{d}
s_i s_j
\operatorname{Tr}\!\left(H|\chi_i\rangle\langle\chi_j|\right)
|\zeta_i\rangle\langle\zeta_j|.
\end{align*}

Since
$
\operatorname{Tr}\!\left(H|\chi_i\rangle\langle\chi_j|\right)
=
\langle\chi_j|H|\chi_i\rangle,
$
we obtain
\[
\sigma(H)
=
\sum_{i,j=1}^{d}
s_i s_j
\langle\chi_j|H|\chi_i\rangle
|\zeta_i\rangle\langle\zeta_j|.
\]

Therefore, the matrix elements of \(\sigma(H)\) in the Schmidt basis
\(\{|\zeta_i\rangle\}\) are
\[
\langle\zeta_i|\sigma(H)|\zeta_j\rangle
=
s_i s_j
\langle\chi_j|H|\chi_i\rangle.
\]
Since the transpose of \(H\) in the Schmidt basis satisfies
$
(H^T)_{ij}
=
\langle\chi_j|H|\chi_i\rangle,
$
it follows that
$
\sigma(H)_{ij}
=
s_i(H^T)_{ij}s_j,
$
or equivalently,
$
\sigma(H)=S H^T S,
$
where \(S=\operatorname{diag}(s_1,\ldots,s_d)\).
\end{proof}

\begin{thm}\label{th5}
Let 
\(
|\psi\rangle_{AB} = \sum_{i=1}^{d} s_i \, |\chi_i\rangle_A \otimes |\zeta_i\rangle_B
\)
be a bipartite pure state with full Schmidt rank (\(s_i>0\) and \(\{|\chi_i\rangle\}_i,\{|\zeta_i\rangle\}_i\) orthonormal).
Let \(P = |a\rangle\langle a|\) and \(Q = |b\rangle\langle b|\) be rank-one projectors on system \(A\), and let \(\rho_B(P), \rho_B(Q)\) be the corresponding normalized reduced states on \(B\) given the outcome corresponding to $P/Q$ is observed on $A$.  
If \(\rho_B(P) = \rho_B(Q)\), then \(P=Q\). Hence distinct rank-one projectors on \(A\) generate distinct normalized reduced states on \(B\).
\end{thm}

\begin{proof}
Let \(U_A\) and \(U_B\) be the unitaries mapping the Schmidt bases to the computational basis:
\[
U_A |i\rangle = |\chi_i\rangle, \quad U_B |i\rangle = |\zeta_i\rangle.
\]
Then 
\[
|\psi\rangle_{AB} = (U_A \otimes U_B) \sum_{i=1}^{d} s_i |i\rangle \otimes |i\rangle = (U_A \otimes U_B) |\Omega_S\rangle.
\]

For any operator \(H\) on \(A\), the unnormalized reduced state on \(B\) is
\[
\sigma(H) = \operatorname{Tr}_A \big[ (H \otimes I) |\psi\rangle\langle \psi| \big]
= U_B \, \operatorname{Tr}_A \big[ (U_A^\dagger H U_A \otimes I) |\Omega_S\rangle\langle \Omega_S| \big] U_B^\dagger.
\]

From the previous lemma, we can infer that 
\[
\operatorname{Tr}_A \big[ (\widetilde{H} \otimes I) |\Omega_S\rangle\langle \Omega_S| \big] = S \widetilde{H}^T S,
\]
where $S=\mathrm{diag}(s_1,\dots,s_d),  \widetilde{H}=U_A^\dagger H U_A$.
Hence
\[
\sigma(H) = U_B \, S \, (U_A^\dagger H U_A)^T \, S \, U_B^\dagger.
\]

Assume the normalized reduced states coincide:
\[
\frac{\sigma(P)}{\operatorname{Tr} \big[ (P \otimes I) |\psi\rangle\langle \psi| \big]} = \frac{\sigma(Q)}{\operatorname{Tr} \big[ (Q \otimes I) |\psi\rangle\langle \psi| \big]}.
\]
Then there exists \(c>0\) such that
\[
\sigma(P) = c\,\sigma(Q) \implies S (U_A^\dagger P U_A)^T S = c\, S (U_A^\dagger Q U_A)^T S.
\]
Since \(S\) is invertible,
$
(U_A^\dagger P U_A)^T = c\, (U_A^\dagger Q U_A)^T.
$

Both sides are rank-one projectors, which cannot be proportional unless \(c=1\) and the vectors differ by a global phase. Therefore
\[
U_A^\dagger P U_A = U_A^\dagger Q U_A \implies P=Q.
\]
\end{proof}

\emph{Proof of Theorem \ref{th_asym}: }  Alice starts the protocol by choosing the suitable product probe $\ket{0}\ket{\eta}$. After her measurements being implemented on the probe, if she finds first outcome is clicking, then it must be first or third measurement. Similarly, for the second outcome, she will know that the possibility is either second or fourth measurement. After the communication of Alice, Bob has to distinguish between $\{\sigma^z_b,\overline{\sigma}^z_b\}$ or $\{\sigma^x_b, \overline{\sigma}^x_b\}$. First set of measurements can be discriminated using $\ket{0}\ket{\eta}$ or $\ket{1}\ket{\eta}$ probe and second set by $\ket{+}\ket{\eta}$ or $\ket{-}\ket{\eta}$ probe. Thus the set of measurements being discriminated locally if Alice is the starter. Note that $\ket{\eta}$ is an arbitrary state because they do not contribute in distinguishability.

Now we want to find a successful protocol taking Bob as the commencing party. As Bob's measurements concerned, one can note that for any outcome, there are four different rank-one projectors acting on the probe. From Theorem \ref{th5}, we can infer that the four reduced states will be four different qubit states for each outcome. As the states spanned dimension two, the highest number of states can be eliminated from the set is one. After communicating to Alice, she has to distinguish between three qubit measurements, which eventually leads to the discrimination of three qubit states which is impossible. So these measurements are not LOCC distinguishable.

\section{Proof of Theorem~\ref{thm:bipartition-wise-indistinguishability}}\label{D}

Before going to the main proof, we present a necessary lemma and theorem which will help to shorten the
complexity of the main proof.

\begin{lemma}[Indistinguishability of rotated product measurements]
\label{lem:rotated-product-measurements}
Let
$$
U_\theta
=
\begin{pmatrix}
\cos\theta & -\sin\theta\\
\sin\theta & \cos\theta
\end{pmatrix},
$$
and consider the two $m$-qubit projective measurements
$\mathcal{Z}^{(m)}
=
\left\{
|\mathbf{x}\rangle\langle\mathbf{x}|
\right\}_{\mathbf{x}\in\{0,1\}^m}$
and
$\mathcal{U}_{\theta}^{(m)}
=
\left\{
U_\theta^{\otimes m}
|\mathbf{x}\rangle\langle\mathbf{x}|
(U_\theta^\dagger)^{\otimes m}
\right\}_{\mathbf{x}\in\{0,1\}^m}.$
If $0<m\theta<\frac{\pi}{2},$
then $\mathcal{Z}^{(m)}$ and
$\mathcal{U}_{\theta}^{(m)}$ cannot be perfectly distinguished with
a single use.
\end{lemma}
\begin{proof}
Assume, to the contrary, that the two measurements can be perfectly
distinguished. By Lemma~\ref{lem:general-unitary-condition}, there
must exist a density operator $\rho$ on the $m$ measured qubits such
that $\operatorname{Tr}
\left(
\rho U_\theta^{\otimes m}
\right) =0.$
Taking the real part yields
\begin{equation}\label{rH=0}
    \operatorname{Tr}(\rho H_m)=0,
\end{equation}
where $ H_m = \frac{1}{2}
\left[U_\theta^{\otimes m} +
(U_\theta^\dagger)^{\otimes m}
\right].$ Since the eigenvectors of a tensor-product operator
are tensor products of the eigenvectors of its factors, an eigenvalue
of $U_\theta^{\otimes m}$ is obtained by multiplying $m$ eigenvalues of $U_\theta$. If $k$ factors contribute $e^{-i\theta}$ and the remaining $m-k$ factors contribute $e^{i\theta}$, then the resulting eigenvalue is $e^{i(m-2k)\theta}$. Therefore, the eigenvalues of $U_\theta^{\otimes m}$ are
$e^{i(m-2k)\theta},\
k=0,1,\ldots,m, $ with multiplicity $\binom{m}{k}$. Consequently, the eigenvalues of $H_m$ are $ \cos[(m-2k)\theta], \ k=0,1,\ldots,m.$
Since $|(m-2k)\theta| \leq m\theta < \frac{\pi}{2},$ we have $\cos[(m-2k)\theta]>0 \ \forall\,k.$ Hence $H_m>0.$
Since $\rho$ is a density operator, it follows that
$\operatorname{Tr}(\rho H_m)>0,$ which contradicts the necessary condition \eqref{rH=0}.
Therefore,
$\mathcal{Z}^{(m)}$ and $\mathcal{U}_{\theta}^{(m)}$ cannot be perfectly distinguished  whenever
$m\theta<\pi/2$.
\end{proof}


\begin{proof}[Proof of Theorem~\ref{thm:bipartition-wise-indistinguishability}]
We first prove global perfect distinguishability. Consider the probe
\begin{equation}
|\Psi_w\rangle =
\frac{1}{\sqrt{2^{w-1}}}
\sum_{\substack{x\in\{0,1\}^w\ |x|\ {\rm even}}}
(-1)^{|x|/2}|x\rangle \ket{\eta},
\end{equation}
where \(|x|=\sum_{i=1}^w x_i\) denotes the Hamming weight of \(x\) and $\ket{\eta}$ arbitrary.

If the unknown measurement is \(\mathcal{G}_{a_1a_2\cdots a_w|1}\), the probe (first part as the measurement acts on it) has support only on the even-parity strings,
$
\mathcal{G}_{a_1a_2\cdots a_w|1}:
\left\{
x\in\{0,1\}^w:\ |x|\ {\rm even}
\right\},
$
with each outcome occurring with probability \(1/2^{w-1}\).

Using
$
|0\rangle=\cos\theta\,|u\rangle+\sin\theta\,|u^\perp\rangle,
|1\rangle=\sin\theta\,|u\rangle-\cos\theta\,|u^\perp\rangle,
$
the probe state can be expanded directly in the
\(\{|u\rangle,|u^\perp\rangle\}^{\otimes w}\) basis. For an outcome containing \(w-r\) copies of \(|u\rangle\) and \(r\) copies of \(|u^\perp\rangle\), its amplitude is
\begin{align}
&\frac{1}{\sqrt{2^{w-1}}}
\sum_{\substack{x\in{0,1}^w\ |x|\ {\rm even}}}
(-1)^{|x|/2}
\prod_{i=1}^{w-r}\langle u|x_i\rangle
\prod_{i=w-r+1}^{w}\langle u^\perp|x_i\rangle
\nonumber\\
&=
\frac{1}{2\sqrt{2^{w-1}}}
\left[
(\cos\theta+i\sin\theta)^{w-r}
(\sin\theta-i\cos\theta)^r
+
(\cos\theta-i\sin\theta)^{w-r}
(\sin\theta+i\cos\theta)^r
\right]
\nonumber\\
&=
\frac{1}{\sqrt{2^{w-1}}}
\cos\left(w\theta-\frac{r\pi}{2}\right).
\end{align}

Thus, for \(\theta=\pi/(2w)\), the amplitude becomes
$
\frac{1}{\sqrt{2^{w-1}}}
\sin\left(\frac{r\pi}{2}\right),
$
which vanishes for even \(r\) and has magnitude \(1/\sqrt{2^{w-1}}\) for odd \(r\).
Hence the only possible outcomes of \(\mathcal{G}_{a_1a_2\cdots a_w|2}\) are the odd-parity strings,
$
\mathcal{G}_{a_1a_2\cdots a_w|2}:
\left\{
y\in{0,1}^w:\ |y|\ {\rm odd}
\right\},
$
each occurring with probability \(1/2^{w-1}\).
Thus, the two measurements have disjoint outcome supports:
$
\operatorname{supp}(\mathcal{G}_{a_1a_2\cdots a_w|1})
\cap
\operatorname{supp}(\mathcal{G}_{a_1a_2\cdots a_w|2})
=
\varnothing.
$
Consequently, the two \(m\)-partite product measurements
\(\mathcal{G}_{a_1a_2\cdots a_w|1}\) and \(\mathcal{G}_{a_1a_2\cdots a_w|2}\) are perfectly distinguishable
globally using the probe state \(|\Psi_w\rangle\).

We now prove that perfect discrimination is impossible across every bipartition. Since both measurements are invariant under permutations
of $A_1,A_2,\cdots,A_w$, it is sufficient to consider the bipartition $A_1|A_2\cdots A_w$. For the first party, the two possible measurements are $\sigma^z_{a_1}$ and $\sigma^u_{a_1}.$ Applying Lemma~\ref{lem:rotated-product-measurements} for these two measurements, taking $m=1$ and $\theta=\pi/2w$, we know that they cannot be perfectly distinguished since $\frac{\pi}{2w}<\frac{\pi}{2}.$

On the joint $A_2\cdots A_w$ side, the two possible measurements are
$\bigotimes_{i=2}^w\sigma^z_{a_i}
 \text{ and } \bigotimes_{i=2}^w\sigma^u_{a_i}.$
Here $m=w-1$, and $ (w-1)\theta = (w-1)\frac{\pi}{2w} < \frac{\pi}{2}$, if $w$ is finite.
Hence, again by
Lemma~\ref{lem:rotated-product-measurements}, these two measurements cannot be perfectly distinguished by any joint probe on $A_2\cdots A_w$. 
It remains to show that adaptive classical communication between
$A$ and $A_2\cdots A_w$ cannot overcome these local obstructions.

Suppose first that $A_1$ acts before the $A_2\cdots A_w$ laboratory. Since no
probe allows perfect discrimination between $\sigma^z_{a_1}$ and $\sigma^u_{a_1}$, $A_1$'s local procedure must possess at least one
branch, say $\lambda$, that occurs with nonzero probability under
both hypotheses and for which the corresponding residual states
$\alpha_{0|\lambda}$ and $\alpha_{1|\lambda}$ are not orthogonal.
Equivalently,
$ F
\left(
\alpha_{0|\lambda},
\alpha_{1|\lambda}
\right)
>0,$
where $F$ denotes the quantum fidelity. Indeed, if every branch common to the two hypotheses produced orthogonal residual states,
Alice could resolve the hypotheses perfectly on every such branch,
contradicting
Lemma~\ref{lem:rotated-product-measurements}.

After receiving the classical message $\lambda$, the joint $A_2\cdots A_w$
laboratory may choose an arbitrary probe depending on $\lambda$.
Nevertheless, it still has to distinguish
$\bigotimes_{i=2}^w\sigma^z_{a_i}
 \text{ from} \bigotimes_{i=2}^w\sigma^u_{a_i}.$ Since these
measurements are not perfectly distinguishable for any probe, there
exists a branch compatible with both hypotheses for which the
resulting states, denoted by $\beta_{0|\lambda}$ and
$\beta_{1|\lambda}$, satisfy
$F\left(\beta_{0|\lambda},
\beta_{1|\lambda}\right)>0.$ Conditioned on this common branch, the states associated with the two
hypotheses are of the form
$
\rho_{0|\lambda}
=
\alpha_{0|\lambda}
\otimes
\beta_{0|\lambda},
\rho_{1|\lambda}
=
\alpha_{1|\lambda}
\otimes
\beta_{1|\lambda}.
$
By multiplicativity of fidelity,
$
F
\left(
\rho_{0|\lambda},
\rho_{1|\lambda}
\right)
=
F
\left(
\alpha_{0|\lambda},
\alpha_{1|\lambda}
\right)
F
\left(
\beta_{0|\lambda},
\beta_{1|\lambda}
\right)
>0.
$
Any further finite sequence of local quantum operations and classical
communication defines a quantum channel $\Lambda$. Since fidelity is
monotonic under completely positive trace-preserving maps,
$
F
\left(
\Lambda(\rho_{0|\lambda}),
\Lambda(\rho_{1|\lambda})
\right)
\geq
F
\left(
\rho_{0|\lambda},
\rho_{1|\lambda}
\right)
>0.
$
The final states therefore cannot have orthogonal supports. Hence
perfect discrimination is impossible.

The same argument applies when the joint  laboratory acts first.
Since $\bigotimes_{i=2}^w\sigma^z_{a_i}
 \text{ and } \bigotimes_{i=2}^w\sigma^u_{a_i}$ cannot
be perfectly distinguished by any joint $A_2\cdots A_w$ probe, there exists a
common branch with nonorthogonal residual states. Conditioned on that
branch, $A_1$ subsequently faces $\sigma^z_{a_1}$ versus $\sigma^u_{a_1}$, which cannot be perfectly distinguished by any probe.
The same multiplicativity and monotonicity argument therefore rules
out perfect discrimination after any finite number of communication
rounds.

Thus $\mathcal{G}_{a_1a_2\cdots a_w|1}$ and $\mathcal{G}_{a_1a_2\cdots a_w|2}$ are not perfectly
distinguishable by finite-round LOCC across $A_1|A_2\cdots A_w$. By permutation
symmetry, the same conclusion holds across the other bipartitions. Therefore, $\mathcal{G}_{a_1a_2\cdots a_w|1}$ and $\mathcal{G}_{a_1a_2\cdots a_w|2}$ are perfectly
distinguishable globally but are not perfectly distinguishable by
finite-round LOCC across any bipartition.
\end{proof}

\end{document}